\documentclass[conference]{IEEEtran}
\IEEEoverridecommandlockouts

\usepackage{amsmath,amssymb,amsfonts}
\usepackage{algorithmic}
\usepackage{graphicx}
\usepackage{textcomp}
\usepackage{xcolor}
\def\BibTeX{{\rm B\kern-.05em{\sc i\kern-.025em b}\kern-.08em
    T\kern-.1667em\lower.7ex\hbox{E}\kern-.125emX}}
\usepackage[colorlinks,citecolor=blue,urlcolor=blue,filecolor=blue]{hyperref}

\usepackage{subcaption}
\usepackage{booktabs}   

\usepackage{amsthm}

\newtheorem{theorem}{Theorem}

\begin{document}

\title{Truly Bayesian Entropy Estimation
\thanks{I.K.\ was supported in part by
the Hellenic Foundation for Research and Innovation (H.F.R.I.) under 
the ``First Call for H.F.R.I.
Research Projects to support Faculty members and Researchers and the 
procurement of high-cost
research equipment grant, project number 1034.''}
}

\author{\IEEEauthorblockN{Ioannis Papageorgiou}
\IEEEauthorblockA{\textit{University of Cambridge} \\
\texttt{ip307@cam.ac.uk}}
\and
\IEEEauthorblockN{Ioannis Kontoyiannis}
\IEEEauthorblockA{\textit{University of Cambridge} \\
\texttt{yiannis@maths.cam.ac.uk}}
}

\maketitle

\begin{abstract}
Estimating the entropy rate of discrete time series is a challenging 
problem with important applications in numerous areas including
neuroscience, genomics, image processing and natural language 
processing. A number of approaches have been developed for 
this task, typically based either on universal data compression algorithms,
or on statistical estimators of the underlying process distribution.
In this work, we propose a fully-Bayesian approach for entropy estimation. 
Building on the recently introduced Bayesian Context Trees (BCT) 
framework for modelling discrete time series as variable-memory 
Markov chains, we show that it is possible to sample directly from 
the induced posterior on the entropy rate. This can be 
used to estimate the entire posterior distribution,
providing much richer information than 
point estimates. We develop theoretical results for 
the posterior distribution of the entropy rate, 
including proofs of consistency and asymptotic normality. 
The practical utility of the method is illustrated on both
simulated and real-world data, where it is found to outperform  
state-of-the-art alternatives. 
\end{abstract}

\smallskip

\begin{IEEEkeywords}
Entropy estimation, Bayesian context trees, 
context-tree weighting, entropy rate, neuroscience
\end{IEEEkeywords}



\section{Introduction}

The task of estimating the entropy rate of a discrete-valued stochastic process from empirical data is an important and challenging problem which dates back to the original work of Shannon on the entropy of English text~\cite{shannon1951prediction}. Since then it has received a lot of attention, perhaps
most notably
in connection with neuroscience \cite{archer2013bayesian,strong1998entropy,london2002information,nemenman2004entropy,paninski2003estimation,bhumbra2004measuring}, where the entropy rate 
plays a crucial role in efforts to
describe and quantify the amount of information communicated
between neurons. Other important areas of applications include genomics~\cite{dna1}, image processing~\cite{hero1999asymptotic}, and web navigation~\cite{levene2003computing}. 

Throughout the years, a very wide variety of approaches have been developed for this task. A recent extensive review of the relevant literature is given in~\cite{verdu2019empirical}. Perhaps the simplest method is the plug-in estimator, which uses the empirical probability estimates 
of blocks to approximate the entropy rate. However, this approach is well-known to perform poorly in practice, mainly 
because of the undersampling problem, i.e., the difficulty in obtaining accurate probability estimates for large block-lengths~\cite{paninski2004estimating,verdu2019empirical,nemenman2004entropy,gao2008estimating}.
A different family of estimators is based on string-matching and the fundamental results of Lempel-Ziv (LZ)~\cite{ziv1977universal,wyner1989some} for the relation of match-lengths with the entropy rate; 
the most effective such estimators are described 
in~\cite{kontoyiannis1998nonparametric} and~\cite{gao2008estimating}.

\newpage 

In view of the Shannon-McMillan-Breiman theorem~\cite{cover1999elements}, it is easy to see that, for a string $x_1 ^n=(x_1,\ldots,x_n)$, any consistent probability estimate $\widehat p_n (x_1 ^n)$ leads to an estimate for the entropy rate as $\widehat H  = -(1/n) \log \widehat p_n (x_1 ^n) $. In this setting, the adaptive probability assignments of universal compressors like 
the context-tree weighting (CTW) algorithm~\cite{ctw} and 
prediction by partial matching (PPM)~\cite{cleary1984data} 
are natural choices for the probability estimates to be employed. 
Finally, a different approach is proposed in~\cite{cai2004universal}, 
based on block sorting~\cite{bwt}.
[Throughout the paper, $\log$ denotes the natural
logarithm.]


The first contribution of this work is the introduction of a new,
fully-Bayesian approach for estimating the entropy rate,
building on the recently introduced Bayesian Context Trees (BCT) 
framework for discrete time series~\cite{our}.
The BCT framework is a Bayesian modelling framework for the 
class of variable-memory Markov chains, which has been found 
to be very effective in a range of statistical tasks 
for discrete time series, including model selection and 
prediction~\cite{our,ourisit} as well as change-point detection~\cite{changepoint,lungu2022bayesian}. Context-trees in data compression were also recently studied in~\cite{matsushima2009reducing,nakahara2022probability}. 

Our second contribution is the introduction of a Monte Carlo algorithm
that produces independent and identically distributed~(i.i.d.) samples
from the induced posterior on the entropy rate.
This facilitates fully-Bayesian estimation:
The Monte Carlo samples can be used to estimate the entire posterior 
distribution of the entropy rate, which contains significantly richer
information than simple point estimates. 

Our third contribution is in stating and proving
provide theoretical results for the asymptotic behaviour 
of this posterior distribution, giving
strong theoretical justifications for the use of our methods. 
We show that the posterior asymptotically concentrates
on the true underlying value of the entropy with probability 1,
and that it is asymptotically normal.

Finally, we present experimental results on both
simulated and real-world data, illustrating
the performance of our methods in practice. 
We provide estimates of the entire posterior of the entropy rate,
as well as specific point estimates
with associated `Bayesian confidence intervals' or
`credible regions'.
Compared with previous approaches,
our proposed point estimates outperform 
state-of-the-art alternatives. 


\section{Background: Entropy Estimation} \label{previous}

The {\em entropy rate} $\bar{H}$ 
of a process $\{X_n \}$ on a finite alphabet
is $\bar{H} = \lim _ {n \to \infty} (1/n)H (X_1^n)$,
whenever the limit exists, where $H(X_1^n)$ denotes the
usual Shannon entropy of the discrete random vector
$X_1^n=(X_1,\ldots,X_n)$. 


\newpage

\noindent
{\bf Plug-in estimator. } 
Motivated by the definition of the entropy rate,
the simplest and one of the most commonly used
estimators of the entropy rate is the per-sample
entropy of the empirical distribution of $k$-blocks.
Letting $\widehat p_k (y_1 ^k ) $, $y_1^k\in A^k$,
denote the empirical distribution of $k$-blocks
induced by the data on $A^k$, the {\em plug-in}
or {\em maximum-likelihood} estimator is simply,
$\widehat H_ k = (1/k) H (\widehat p_k)$.
The main advantage of this estimator
is its simplicity. Well-known drawbacks
include its high variance due to undersampling,
and the difficulty
in choosing appropriate block-lengths $k$
effectively. Another limitation is its always negative bias, 
with a number of different approaches developed to 
correct for it; 
see~\cite{verdu2019empirical,gao2008estimating,paninski2003estimation} 
for more details.  

\medskip

\noindent
{\bf Lempel-Ziv estimators. } 
Among the numerous match-length-based entropy estimators
that have been derived from the Lempel-Ziv family of 
data compression algorithms, the
increasing-window estimator of \cite{gao2008estimating}, has been
identified as the most effective one,
and it will be used in our comparisons.

\medskip

\noindent
{\bf Naive CTW estimator. } This uses the marginalised probability estimate
$P_{\rm CTW}(x_1^n)$ computed by the CTW algorithm,
to define $\widehat H _ { \text {CTW}} = -(1/n)\log P_{\rm CTW}(x_1^n)$. 
This estimator was found 
in \cite{gao2008estimating,verdu2019empirical} to 
achieve good performance in practice. 
Its consistency and asymptotic normality follow easily from standard
results, and its (always positive) bias is of $O((\log n)/n)$.

\medskip

\noindent
{\bf PPM estimator. } Using a different
adaptive probability assignment, $Q(x_1^n)$, this method
forms the entropy estimate
$\widehat H _ { \text {PPM}} = -(1/n)\log Q(x_1^n)$,
where prediction by partial matching is used to fit the model 
that leads to~$Q(x_1^n)$. 
Here we use the interpolated smoothing variant of PPM 
introduced by \cite{bunton1996line}.

\medskip

\section{Bayesian Context Trees} \label{previous_bct}

The BCT model class consists of
variable-memory Markov chains, where
the memory length of the process may
depend on the
most recently observed symbols. 
Variable-memory chains 
admit natural representations as context trees. 
Let $\{X_n\}$ 
be a $d$th order Markov chain,
with values in the alphabet 
$\mbox{$A=\{0,1,\ldots,m-1\}$}$.
The {\em model} describing $\{X_n\}$ 
as a variable-memory chain is represented by
a proper $m$-ary tree $T$ as the one shown
below,
where $T$ is {\em proper} if any node in $T$ 
that is not a leaf has exactly $m$ children. 
At each leaf $s$ there is an associated parameter 
$\theta_s$ given by a probability vector,
$\theta_s=(\theta_s(0),\ldots,\theta_{s}(m-1))$.
The conditional distribution of $X_n$
given the past~$d$ observations $(x_{n-1}, \ldots , x_{n-d})$, is 
given by the vector $\theta_s$ associated to the unique leaf $s$ of $T$ 
that is a suffix of $(x_{n-1}, \ldots , x_{n-d})$. 

\medskip

\centerline{\includegraphics[width= 0.7 \linewidth]{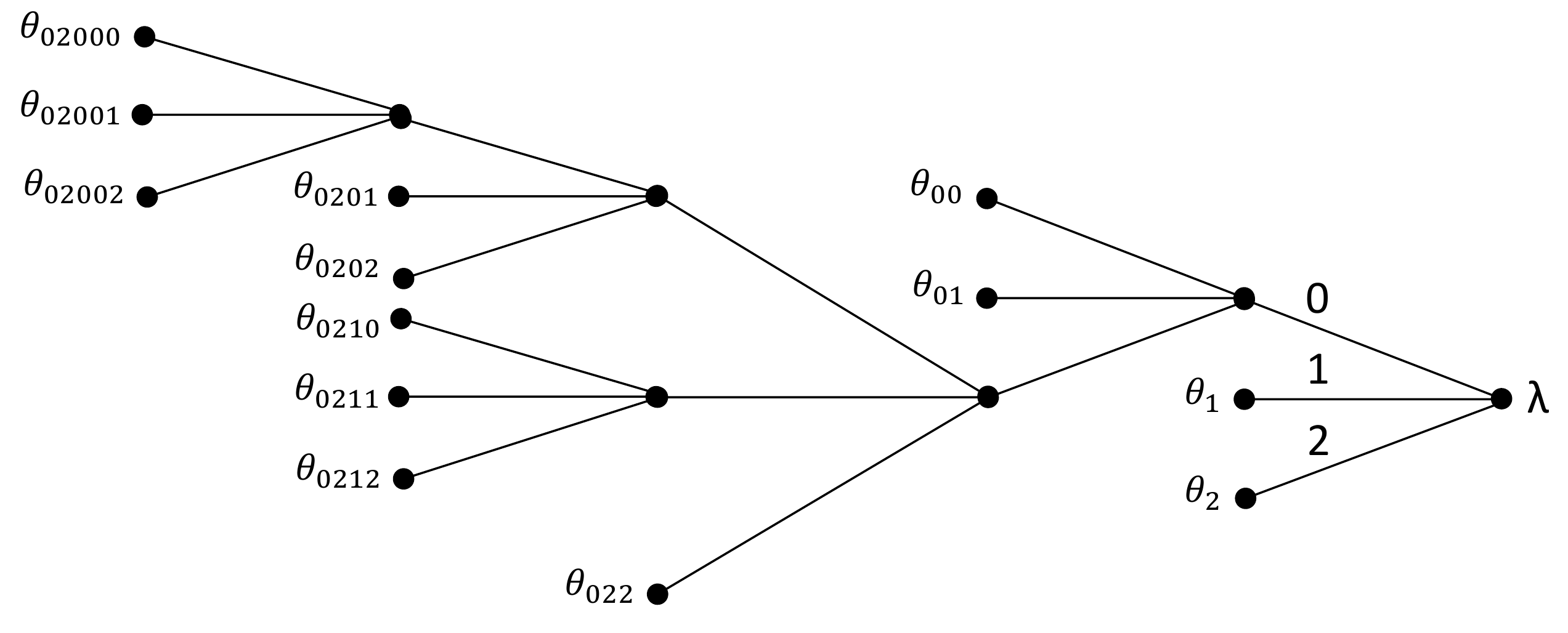}}

\noindent
{\bf Model prior. } 
Let $\mathcal {T} (D) $ denote the 
collection of all proper $m$-ary trees 
with depth no greater than $D$.
Following \cite{our}, we use 
the following prior distribution 
on $\mathcal {T} (D)$,
\begin{equation}
\label{prior_T}
\pi(T)=\pi_D(T;\beta)=\alpha^{|T|-1}\beta^{|T|-L_D(T)} \ ,
\end{equation}
where $\beta\in(0,1)$ is a hyperparameter, 
$\alpha=(1-\beta)^{1/(m-1)}$,
$|T|$ is the number of leaves of $T$,
and $L_D(T)$ is the number of leaves of 
$T$ at depth $D$.
We adopt the default value of
$\beta = 1 - 2 ^ {-m+1}$~\cite{our}.

\medskip

\noindent
{\bf Prior on parameters. } Given a model $T \in \mathcal{T} (D)$, an independent Dirichlet prior
with parameters $(1/2,1/2,\ldots,1/2)$ is placed
on each $\theta_s$, $s\in T$.

\medskip

\noindent
{\bf CTW: The context tree weighting algorithm. }
Given observations $x_1^n$ with an initial context
$x_{-D+1}^0$:

1) Build the tree $T_{\text{MAX}}$, 
which is the smallest proper tree that contains
all the contexts $x_{i-D+1} ^ {i}, \  i=1,2,\ldots,n$,
as leaves. 

2) Compute the {\em estimated probabilities}
$P_{e,s}$,
\begin{equation} 
P_{e,s} = 
\frac{\prod_{j=0}^{m-1} 
[(1/2)(3/2)\cdots 
(a_s(j)- 1/2)]}
{(m/2)(m/2+1) 
\cdots (m/2+M_s-1)},
\label{eq:Pe}
\end{equation} 
for each node $s$ of $T_{\text{MAX}}$;
the \textit{count vectors}
$a_s=(a_s(0),a_s(1),\ldots,a_s(m-1))$
are given by
$a_s(j)= \#$  times symbol $j\in A$ follows 
context~$s$ in $x_1^n$,
and $M_s=a_s(0)+a_s(1)+\cdots+a_s(m-1)$.

3) Starting at the leaves and proceeding recursively towards the root, for each node~$s$ of $T_{\text{MAX}}$ compute the \textit{weighted probabilities} $P_{w,s}$, given by, 
\begin{equation} 
\label{pws}
P_{w,s}\!=\!
\left\{
\begin{array}{ll}
P_{e,s},   &\mbox{$s$ is a leaf,}\\
\beta  P_{e,s}+(1-\beta)  \prod_{j=0}^{m-1} P_{w,sj},  &\mbox{otherwise.}
\end{array}\!\!
\right.\!\!
\end{equation} 

\noindent
{\bf Prior predictive likelihood. }
As shown in \cite{our}, 
the weighted probability at the root $P_{w,\lambda}$ 
is equal to the {\em prior predictive likelihood}, 
namely, the probability of $x_1^n$ averaged over
both models and parameters:
\begin{equation}
P(x)=\sum_{T\in\mathcal {T}(D)}
\int_\theta P(x|T,\theta)\pi(\theta|T)\pi(T)d\theta,
\label{eq:ppl}
\end{equation}
where we omit the sub-/super-scripts of $x=x_{-D+1}^n$ for 
clarity.

\medskip

\noindent
{\bf Sampling from the joint posterior on $(T, \theta)$. }
The following procedure~\cite{branch_arxiv,branch_isit}
can be followed to produce i.i.d.\ 
samples $T ^ {(i)}$ from the model posterior $\pi(T|x)$. 
For any context~$s$, 
let $P_{b,s} = \beta P_{e,s} / P_{w,s}$, 
and $\lambda$ be the root node.

Let $T=\{\lambda\}$; if $D=0$, stop; 
if $D>0$, then, with probability $P_{b,\lambda}$, 
mark the root as a leaf and stop,
or, with probability $(1-P_{b,\lambda})$,
add all~$m$ children of $\lambda$ at depth $1$ to~$T$. 
If $D=1$, stop;
otherwise, examine each of
the $m$ new nodes and either
mark a node $s$ as a leaf with probability $P_{b,s}$,
or add all~$m$ of its children to $T$ with probability
$(1-P_{b,s})$, independently from node to node. 
Proceeding recursively, at each
step examining all non-leaf nodes at depths strictly
smaller than $D$ until no more eligible nodes remain,
produces a random tree $T \in \mathcal {T} (D)$, 
with probability $\pi(T|x)$.

Next, since the
full conditional density of 
the parameters is $\pi(\theta |T,x) = \prod _ {s \in T} \text{Dir}
\left (1/2 + a_s (0), \ldots,1/2 + a_s (m-1) \right ),$ 
for each $T^{(i)}$
we can draw a conditionally independent sample 
$ \theta ^ {(i)}$ from $ \pi ( \theta | T ^ {(i)},x  )$,
producing a sequence of
exact i.i.d.\ samples $(T ^ {(i)}, \theta ^ {(i)})$ from 
the joint posterior $\pi(T,\theta|x)$. 


\section{The Bayesian Entropy Estimator} \label{entropy}

Our starting point is the following observation.
Suppose $\{X_n\}$ is an ergodic 
variable-memory chain with model $T$
and parameters 
$\theta=\{\theta_s;s\in T\}$. Viewing~$\{X_n\}$
as a full $D$th order chain, its entropy rate can be 
written in terms of its
transition probabilities and its unique stationary
distribution~\cite{cover1999elements}. 
That is,
$\bar{H}$ can be expressed as an explicit 
function $\bar{H}=H(T ,\theta)$.
Therefore, given a time series $x$,
using the sampler of the previous section 
to produce samples $(T^{(i)},\theta^{(i)})$
from $\pi(T,\theta|x)$, we can obtain i.i.d.
samples $H^{(i)}=H(T^{(i)},\theta^{(i)})$ from the
posterior $\pi(\bar{H}|x)$ of the entropy rate.

The calculation of each 
$H^{(i)}=H(T^{(i)},\theta^{(i)})$ is straightforward
and only requires the computation of the stationary
distribution $\pi$ of the induced first-order chain that
corresponds to taking blocks of size [depth$(T^{(i)})+1$].
The only potential difficulty is if either the
depth of $T^{(i)}$ or the alphabet size $m$ 
are so large that the computation of $\pi$ becomes
computationally expensive.
In such cases, $H^{(i)}$
can be computed approximately by including an 
additional Monte Carlo step: Generate a sufficiently
long random sample $Y_{-D+1}^M$ from the chain 
$(T^{(i)},\theta^{(i)})$, 
and calculate:
\begin{equation}
H^{(i)}\approx-\frac{1}{M}\log P(Y_1^M|Y_{-D+1}^0, T^{(i)},\theta^{(i)}).
\label{eq:MCMCMC}
\end{equation}
The ergodic theorem and the central limit theorem
for Markov chains 
\cite{chung1967markov,meyn2012markov}
then guarantee the accuracy
of~(\ref{eq:MCMCMC}).




Our first theoretical result
shows that the posterior of the entropy rate
asymptotically
concentrates on the true underlying value.
Theorem~\ref{thm2} shows that it is asymptotically normal. 
In order to avoid inessential technicalities, 
we assume throughout that the data-generating
process $\{ X_n \}$ is a {\em positive ergodic} 
variable-memory chain with memory no greater than~$D$,
i.e., that all the parameters $\theta _s (j)$ are strictly positive.

\begin{theorem}
\label{thm1} 
Let $X_{-D+1}^n$ be
generated
by a positive-ergodic, variable-memory chain 
$\{ X_n \}$ 
with model $T^*\in\mathcal {T}(D)$, parameters $\theta ^ *$, and entropy rate $H ^ * = H(T^*, \theta ^ *)$; let $\beta$ be arbitrary. The posterior distribution of the entropy rate $\pi( \bar H | X_{-D+1}^n) $ concentrates around the true value $H ^ *$, i.e., 
\begin{equation}
    \pi(\cdot | X_{-D+1}^n) \xrightarrow[]{\mathcal {D}} \delta _ {H ^*} , \; \; \; \;  \text {a.s., as} \  n \to \infty, 
\end{equation}
\textit{%
where $\xrightarrow[]{\mathcal {D}}$ denotes weak convergence of probability measures, and $\delta _ {H ^*} $ is the unit mass at $H^*$. }
\end{theorem}
\begin{proof}
    This is an immediate consequence of the concentration of the joint posterior $\pi(T, \theta |X_{-D+1}^n) $ around the true values $(T ^*, \theta ^* )$, which is proven as Theorem 3.7 in~\cite{bct_theory}.
\end{proof}

\begin{theorem}
    \label{thm2}
Under the assumptions of Theorem~1, let $H ^ {(n)}$ be distributed
according to the posterior distribution $\pi( \bar H | X_{-D+1}^n) $, and let $ \bar \theta_{T^*}  ^ {(n)} $ denote the mean of the posterior of the parameters $\pi (\theta | T ^* ,X_{-D+1}^n )$. Then, as $n \to \infty$,
\begin{equation}
    \sqrt n \big ( H ^ {(n)} - \widetilde H ^{(n)} \big )  \xrightarrow[]{\mathcal {D}} Z \sim \mathcal {N} (0, \sigma _ H ^ 2) , \; \; \; \;  a.s.,
\end{equation}
where $\widetilde H ^{(n)} = H (T ^* ,  \bar \theta_{T^*}  ^ {(n)}) \to H ^*$, a.s. as $n \to \infty$. 
\end{theorem}

\begin{proof}
    First,  denoting $x = X_{-D+1}^n $ for simplicity, for the entropy rate posterior density $\pi( \bar H | x)$ we can write,
    \begin{align} \label{post_exp}
      & \pi( \bar H | x) = \sum _ {T  \in \mathcal {T}(D)} \pi (\bar H | T,  x) \pi(T |x ) ,
    \end{align}
    where for every given $T$, the conditional density $ \pi (\bar H | T,  x)$ arises from the conditional density of the parameters $ \pi (\theta | T,  x)$ through the transformation of variables $\bar H = H (T,\theta)$.

    Denoting the posterior density $\pi _ {\bar H |x} (h)$ from~(\ref{post_exp}) as $f_{\bar H , n }(h)$, in order to prove the theorem it suffices to show that,
    \begin{equation} \label{density_prove}
        \frac{1}{\sqrt {n}} f_{\bar H , n }\left ( \frac {h}{\sqrt {n}} + \widetilde H ^ {(n)}\right ) \to \phi _ {\sigma _ H} (h) , \; \; \text{as} \  n \to \infty,
    \end{equation}
    uniformly on compact sets in $\mathbb R$, where $\phi _ \sigma (z)$ denotes the density of a zero mean Gaussian with variance $\sigma ^ 2$. 

    From the central limit theorem (CLT) on the parameters, proven as Theorem~3.8 in~\cite{bct_theory}, we have that, for $\theta_{T ^*} ^ {(n)}$ distributed according to the posterior $\pi (\theta | T ^* ,x)$, as $ n \to \infty$,
\begin{equation}
    \sqrt n \big (  \theta_{T ^*} ^ {(n)} - \bar \theta_{T^*} ^ {(n)} \big )  \xrightarrow[]{\mathcal {D}} Z \sim \mathcal {N} (0, J) , \; \; \; \;  a.s.,
\end{equation}
where the covariance matrix $J$ is given in~\cite{bct_theory}. Similarly, we can get for every $T$ that,
\begin{equation}
    \sqrt n \big (  \theta_{T }^ {(n)} - \bar \theta_{T} ^ {(n)}  \big )  \xrightarrow[]{\mathcal {D}} Z \sim \mathcal {N} (0, J _T) , \; \; \; \;  a.s.,
\end{equation}
where $\theta_{T} ^ {(n)}$ is distributed according to the posterior $\pi (\theta | T ,x)$ with mean $ \bar \theta_{T} ^ {(n)}$. Now, given $T$, $\bar H$ is just a function of the parameters $\theta$, so applying the delta method~\cite{van2000asymptotic} we get that,
\begin{equation} \label{h_clt}
    \sqrt n \big (  H_{T }^ {(n)} - \widetilde H _T ^{(n)} \big )  \xrightarrow[]{\mathcal {D}} Z \sim \mathcal {N} (0, \sigma _ T ^ 2) , \; \; \; \;  a.s.,
\end{equation}
where  $\widetilde H _T ^ {(n)}= H (T, \bar \theta_{T} ^ {(n)})$,  and $ H_{T }^ {(n)}$ is distributed according to $ \pi (\bar H | T,  x)$. We note that in order to use the delta method, implicitly we are using the differentiability of $H$, which reduces to the differentiability of the stationary distribution $\pi$ as a function of $\theta$; this condition is satisfied for a positive-ergodic chain, see, e.g.~\cite{golub1986using,schweitzer1968perturbation,funderlic1986sensitivity,meyer1975role} and the references therein for more details.

Considering the conditional density $f_{\bar H | T , n}$ of the random variables in~({\ref{h_clt}}), we get that, uniformly on compact sets,
\begin{equation} \label{h_density_t}
\frac{1}{\sqrt {n}} f_{\bar H | T, n }\left ( \frac {h}{\sqrt {n}} + \widetilde H _T ^ {(n)}\right ) \to \phi_{ \sigma_{ T } ^ {} } (h)  , \; \; \text{as} \  n \to \infty.
\end{equation}
Finally, from Theorem 3.6 in~\cite{bct_theory}, the model posterior concentrates on $T^*$, i.e., $\pi(T^*|x) \to 1$ and $\pi(T|x) \to 0$ for~$T \neq T^*$, a.s. as $n\to \infty$.
Combining this with~(\ref{h_density_t}) and~(\ref{post_exp}) gives~(\ref{density_prove}) with $\sigma _H = \sigma _ {T^*}$ and $\widetilde H ^{(n)} = \widetilde H _{T^ *} ^ {(n)}$, completing the proof.
\end{proof}

\noindent
\textbf{Remarks.} 1. The delta method also gives an expression for the asymptotic variance $\sigma_H ^2 = \nabla _{\theta} H ^ { \text T}  J   \nabla _{\theta} H $. This involves the covariance matrix $J$~\cite{bct_theory} and the partial derivatives ${\partial H} /  {\partial \theta }$, which in turn 
involve the partial 
derivatives of the stationary distribution with respect to $\theta$
\cite{golub1986using,schweitzer1968perturbation}. 

2. Theorem~\ref{thm2} shows that the posterior $\pi (\bar H ,x)$ is asymptotically normal around $\widetilde H ^ {(n)}$. 
A CLT also holds for the convergence of $\widetilde H ^ {(n)}$ to the true value
$H ^*$, so it is  easy to see that $ H ^ {(n)} = H^* + O _p (1/\sqrt{n})$.

Finally, we give some theoretical results for the naive CTW entropy 
estimator $\hat H _{\text{CTW}}$ described in Section~\ref{previous}. 
Although this was used in~\cite{gao2008estimating} only for binary data, 
here we prove its consistency and asymptotic normality for finite
alphabets.

\begin{theorem} \label{thm3}
    Under the assumptions of Theorem 1, the naive CTW entropy estimator is consistent, i.e.,
    \begin{equation}
        \hat H  _ {\rm CTW} = - \frac{1}{n} \log  P_ {\rm CTW}(x_1 ^ n) \to H ^ *, \; \; \text {a.s., as} \  n \to \infty,
    \end{equation}
    where $P_ {\rm CTW}(x_1 ^ n)$ is the prior predictive 
	likelihood in~{\em (\ref{eq:ppl})}.
\end{theorem}

\begin{theorem} \label{thm4}
    Under the assumptions of Theorem 1, the naive CTW entropy estimator is asymptotically normal: As $n \to \infty $,
    \begin{equation}
       \sqrt {n} \left ( \hat H  _ {\rm CTW} - H ^* \right ) \to Z \sim \mathcal N (0, \sigma _ {\rm CTW} ^ 2), \; \; \text {a.s.}
    \end{equation}
\end{theorem}

\begin{proof}
    Denoting the likelihood of the data $P (x_1 ^n | T ^*, \theta ^ *)$, both theorems follow immediately from the observation that,
    \begin{equation}
        \log P _ {\text {CTW}}(x_1 ^ n ) = \log P (x_1 ^n | T ^*, \theta ^ *) + O (\log n), \; \; \text{a.s.},
    \end{equation}
    which follows directly from the explicit upper and lower bounds on $P_ {\text{CTW}}(x_1 ^ n)$ given as Theorem~3.1 in~\cite{bct_theory} and ``Barron's lemma" in~\cite{kontoyiannis1997second}. The ergodic theorem and the CLT for Markov chains~\cite{chung1967markov,meyn2012markov} then directly give Theorems~\ref{thm3} and~\ref{thm4}, respectively. Note that \textit{positive} ergodicity is not strictly needed for these results, we just need the ergodic theorem and CLT to hold respectively; see~\cite{chung1967markov,meyn2012markov} for general~conditions. 
\end{proof}

\section{Experimental Results} \label{results}

In this section, the BCT entropy estimator 
(with $D=10$) is compared with 
state-of-the-art approaches
as identified by~\cite{verdu2019empirical} and summarised in Section~\ref{previous}.
The BCT estimator is found to give the most reliable estimates 
on a variety of simulated and real-world
data. Moreover, compared 
to most existing approaches that give simple point estimates 
(sometimes with confidence intervals), 
it has the additional advantage 
that it provides the entire posterior distribution
$\pi(\bar{H}|x)$. Throughout this section, 
all entropies are expressed in {\em nats}.

\medskip

\noindent
{\bf A ternary chain. } We consider
$n=1000$ observations generated from the 5th order,
ternary chain in the example shown in Section~\ref{previous_bct};
see~\cite{branch_arxiv} for the complete 
specification of the associated parameters. 
The entropy rate of this chain is~$\bar{H}=1.02$. 
In Figure \ref{post_evolution} we show estimates of the
prior distribution $\pi(\bar{H})$,
and of the posterior~$\pi(\bar{H}|x)$
based on $n=100$ and $n=1000$ observations.
With $n=1000$, the posterior is
close to a Gaussian  
with mean $\mu=1.005$ and standard deviation 
$\sigma = 0.017$. 
For each histogram $N=10^5$ Monte Carlo samples were 
used, and in each
case (and in all subsequent examples), the vertical 
axis of the histograms shows the frequency of the bins
in the Monte Carlo sample.

\begin{figure}[!ht]
\vspace*{-0.1 cm}
\begin{subfigure}{0.32 \linewidth}
\hspace*{-0.1in}
\includegraphics[width= 1.09 \linewidth]{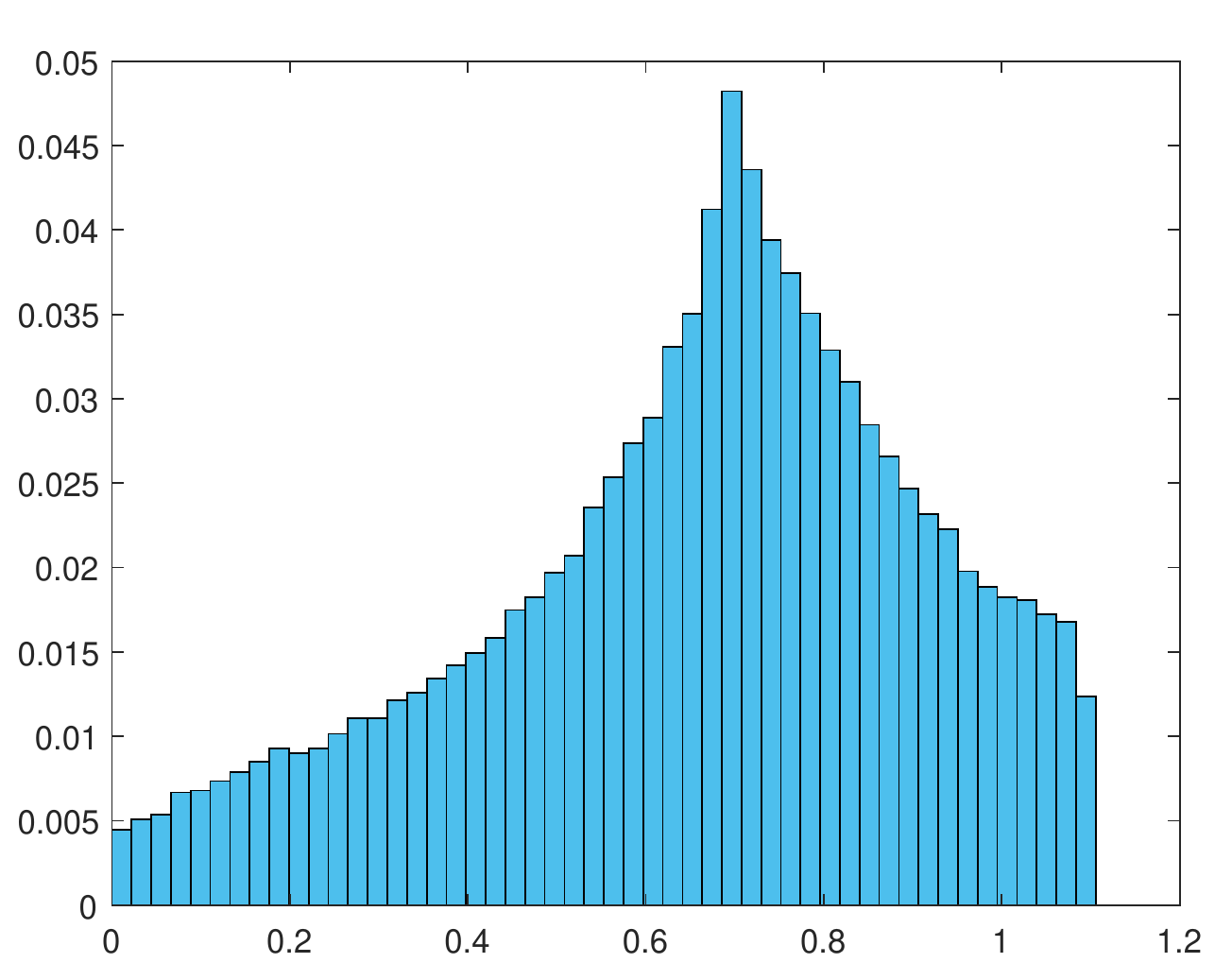}
\vspace*{-0.55 cm}
\caption{prior}
\end{subfigure}
\begin{subfigure}{0.32 \linewidth}
\hspace*{-0.02in}
\vspace*{-0.08 cm}
 \includegraphics[width= 1.10 \linewidth]{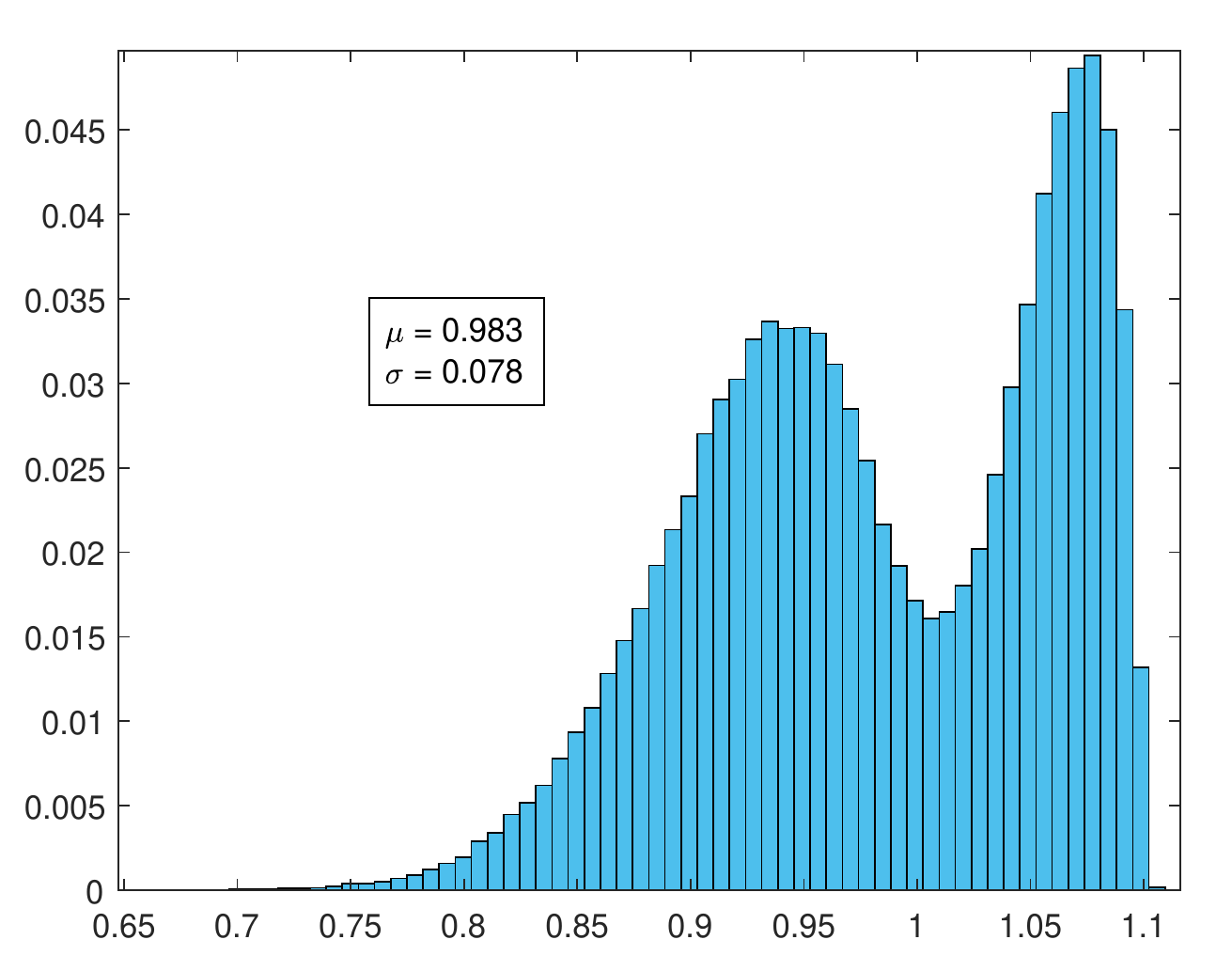}
\vspace*{-0.52 cm}
\caption{$n=100$}
\end{subfigure}
\begin{subfigure}{ 0.32 \linewidth}
\hspace*{0.06in}
 \includegraphics[width= 1.05 \linewidth]{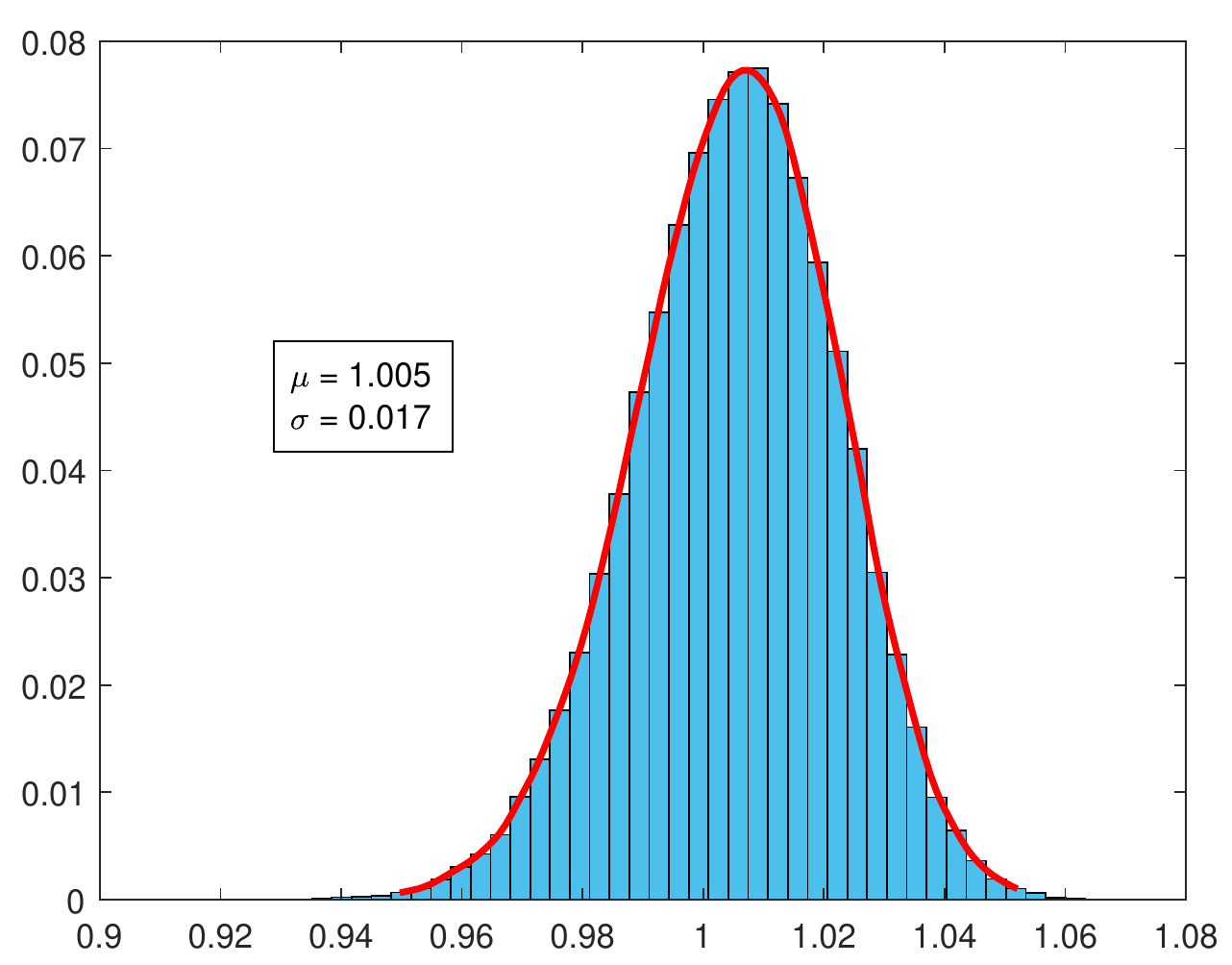}
\vspace*{-0.55 cm}
\caption{$n=1000$}
\end{subfigure}
\vspace*{-0.1 cm}
\caption{Prior $\pi(\bar{H})$ and posterior 
$\pi(\bar{H}|x)$ of the 
entropy rate~$\bar{H}$ with $n=100$ and $n=1000$ observations 
$x$.}
\label{post_evolution}
\end{figure}

Figure \ref{tern_h_plots} shows the performance
of the BCT estimator compared with the other
estimators described above, as a function
of the length $n$ of the available observations~$x$. 
For BCT we plot the posterior mean. For the plug-in
we plot estimates with block-lengths $k=5,6,7$. It is easily observed that the BCT estimator outperforms all the alternatives, and converges faster and closer to the true value of $ \bar H  $.

\begin{figure}[!h]
\centering
\vspace*{-0.3 cm}
\includegraphics[width= 0.77 \linewidth]{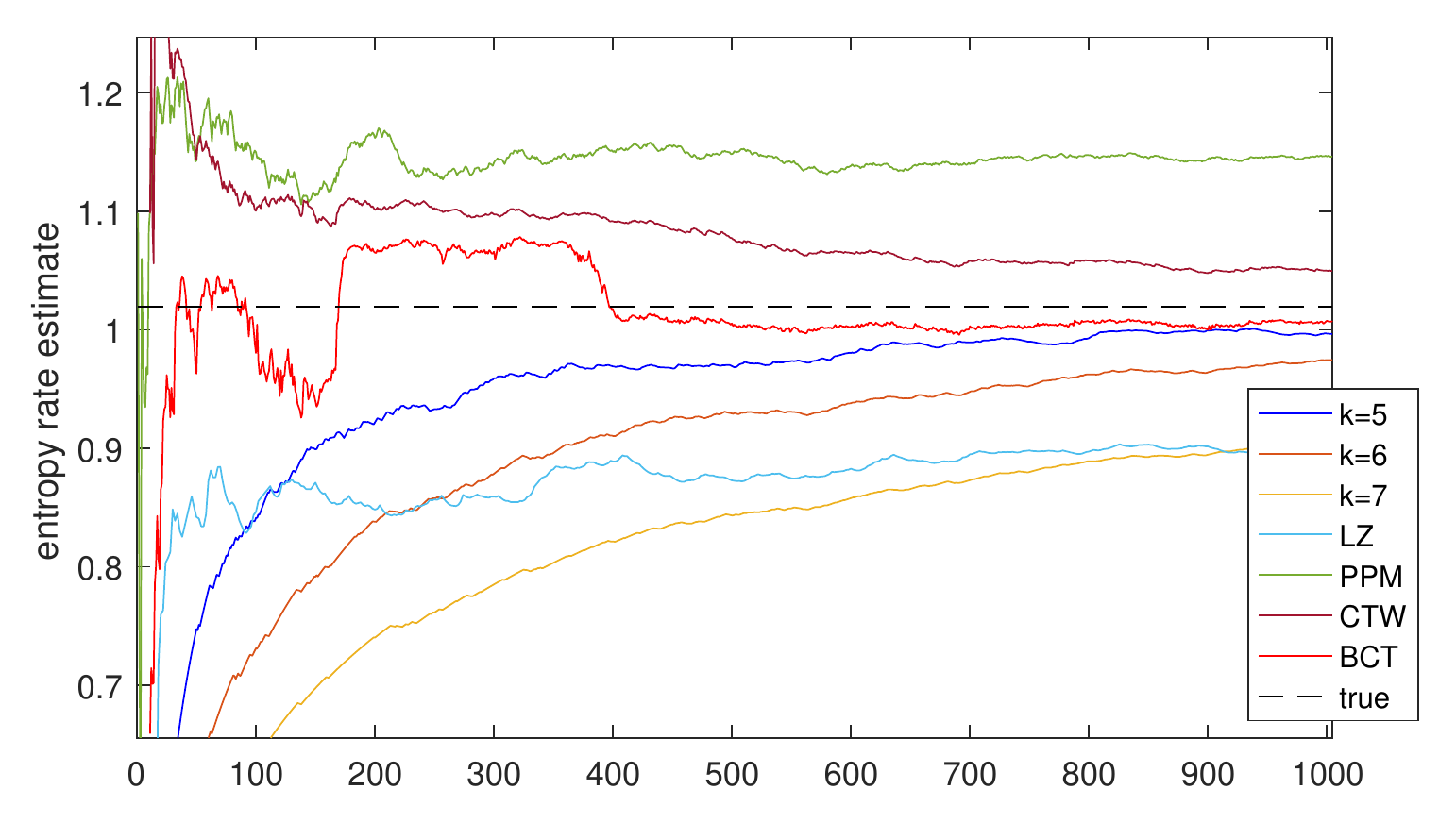}
\vspace*{-0.35 cm}
\caption{Entropy rate estimates for the 5th order ternary chain,
as the number of observations increases.}
\label{tern_h_plots} 
\vspace*{-0.55 cm}
\end{figure}

\noindent
{\bf A third order binary chain. } Here, we consider $n=1000$ observations 
generated from an example of a third order binary chain 
from \cite{berchtold2002mixture}. The underlying model is the 
complete binary tree of depth $3$ pruned at node $s=11$;
the tree model $T$ and the parameter values
$\theta=\{\theta_s;s\in T\}$ are given in~\cite{branch_arxiv}. 
The entropy rate of this chain is $\bar{H}=0.4815$. 
Figure~\ref{y3_pots} shows the performance 
of all estimators,
where 
BCT (using the posterior mean again) is found to have the best performance.
The histogram of the BCT posterior after $n=1000$ observations is
close to a Gaussian with
mean $\mu = 0.4806$ and standard deviation $\sigma = 0.0405$.

\begin{figure}[!h]
\centering
\vspace*{-0.25 cm}
\includegraphics[width= 0.75 \linewidth]{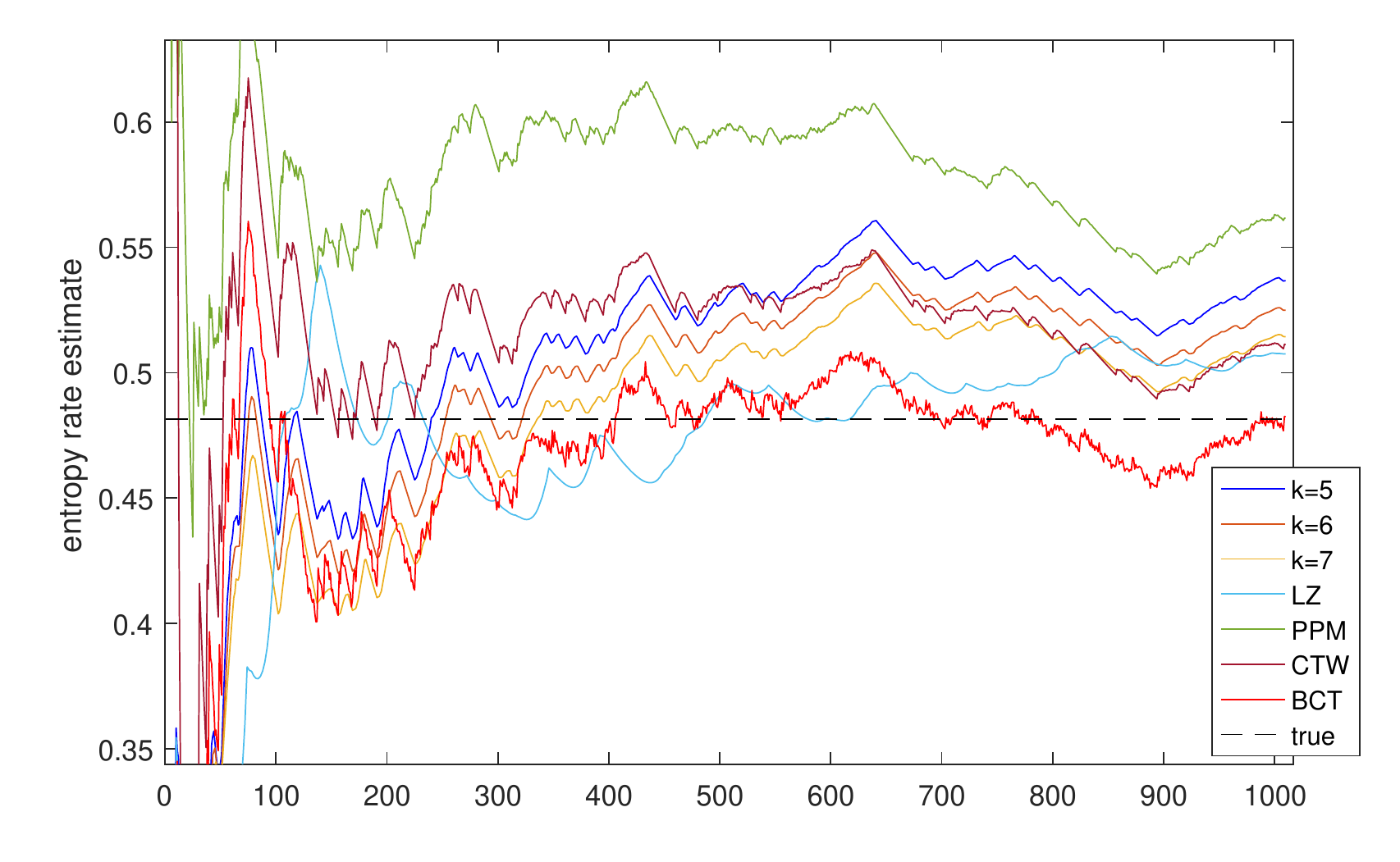}
\vspace*{-0.35 cm}
\caption{Entropy rate estimates for the third order binary chain,
as the number of observations increases.}
\label{y3_pots} 
\vspace*{-0.25 cm}
\end{figure}


\noindent
{\bf A bimodal posterior. } 
We examine a simulated time series
from \cite{our},
consisting of
$n=1450$ observations generated from 
a $3$rd order chain with 
alphabet size $m=6$ and with the property that each $X_n$ depends 
on past observations only via~$X_{n-3}$.
The complete
specification of the chain can be found in~\cite{our,branch_arxiv};
its entropy rate is $\bar{H}=1.355$. 
An interesting aspect of this data set is that 
the model posterior is bimodal, with one mode
corresponding to the empty tree (describing i.i.d.\ observations)
and the other consisting of tree models of depth~3.
As shown in Figure \ref{bimodal_h}, the posterior
of the entropy rate is also bimodal, 
with two approximately-Gaussian modes 
corresponding to the two model posterior modes.

The dominant mode is the one corresponding to models of depth 3; 
it has mean $\mu_1=1.406$, standard deviation $\sigma_1 =0.031$, 
and relative weight $w_1 = 0.91$. The second mode 
corresponding to the empty tree has mean $\mu_2 = 1.632$, standard 
deviation $\sigma_2 = 0.020$, and a much smaller weight 
$w_2 =1 -w_1 =  0.09$. In this case,
the mode of $\pi(\bar{H}|x)$ gives a more reasonable 
choice for a point estimate than the posterior mean.
Like in the previous two examples,
the BCT estimator 
performs better than most benchmarks.

\begin{figure}[!h]
\vspace*{-0.25 cm}
\begin{subfigure}{0.49 \linewidth}
 \includegraphics[width= 1 \linewidth]{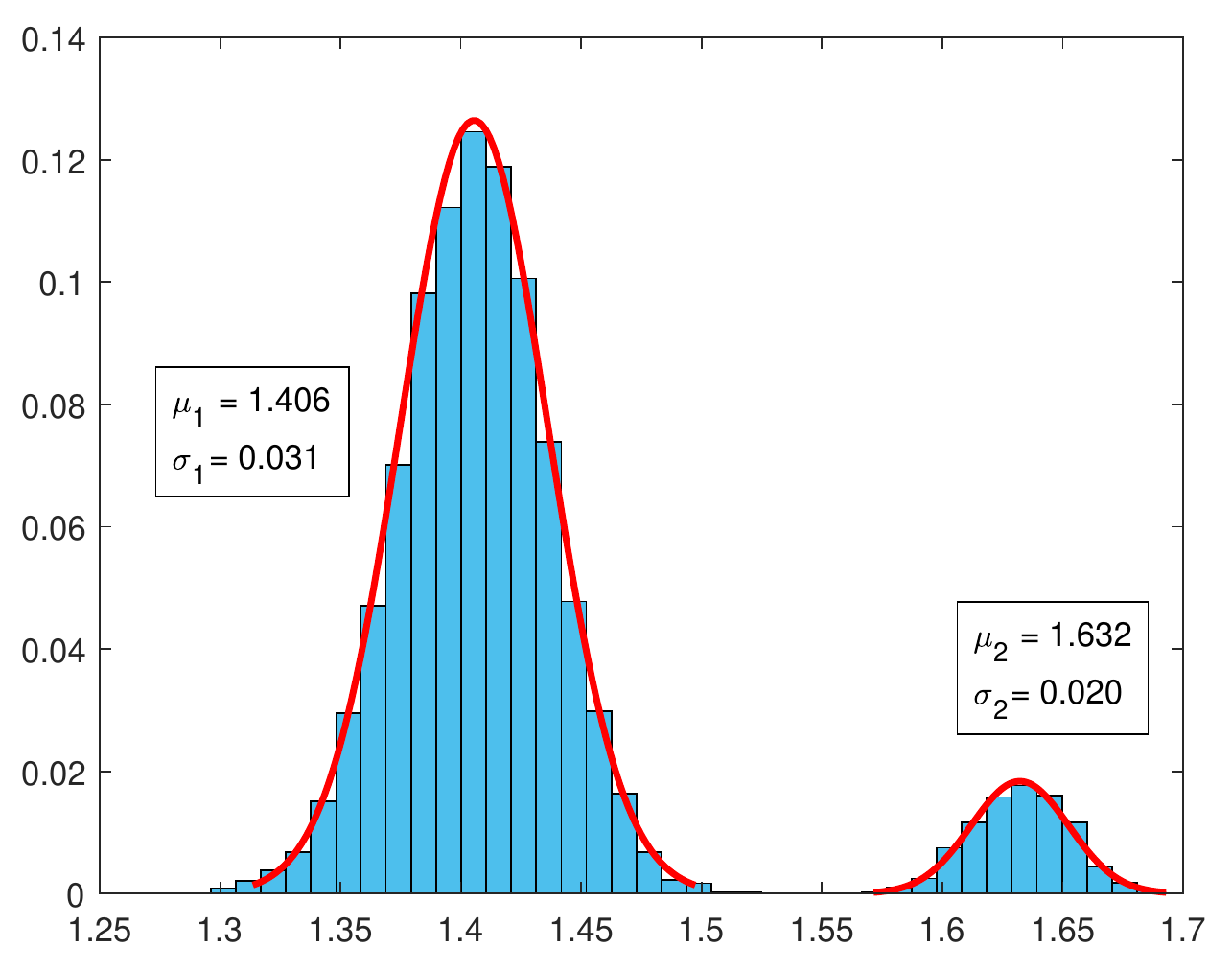}
\vspace*{-0.6 cm}
\caption{bimodal example}
\vspace*{0.02 cm}
\label{bimodal_h}
\end{subfigure}
\begin{subfigure}{0.49 \linewidth}
 \includegraphics[width= 1 \linewidth]{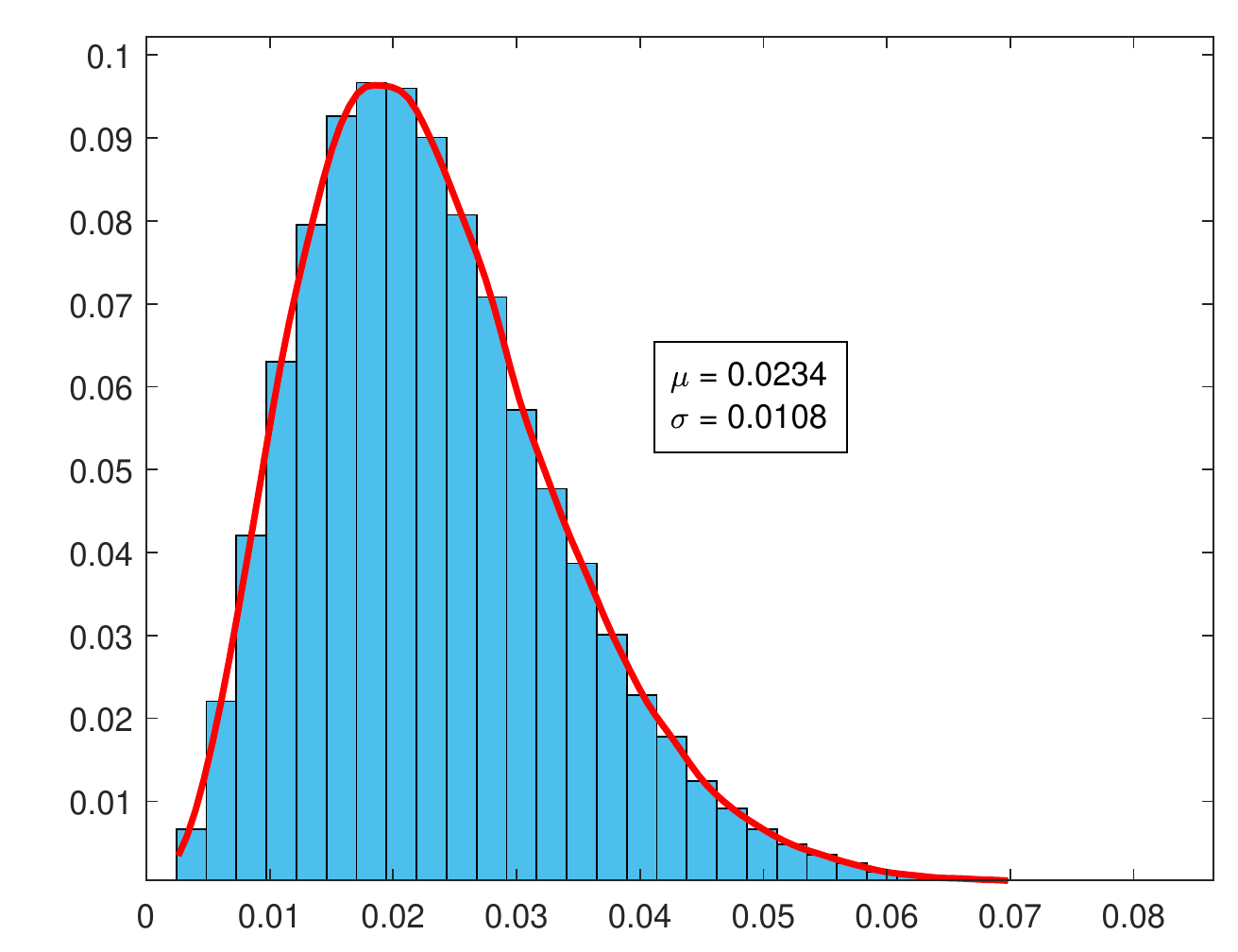}
\vspace*{-0.6 cm}
\caption{spike train}
\label{spike_h}
\end{subfigure}
\vspace*{-0.15 cm}
\caption{Histograms of the
posterior distribution $\pi(\bar{H}|x)$ of the entropy rate,
constructed from $N=10^5$ i.i.d.\ samples.}
\vspace*{-0.20 cm}
\end{figure}

\noindent
{\bf Neural spike train. } 
We consider $n=1000$ binary observations from a spike train recorded from a single neuron in region V4 of a monkey's brain. The BCT posterior is 
shown in Figure \ref{spike_h}: Its mean is $\mu = 0.0234$, its standard deviation 
is $\sigma = 0.0108$, and is skewed to the right. 

This dataset is 
the first part of a long spike train of length $n=3,919,361$ 
from \cite{gregoriou2009high,gregoriou2012cell}.
Although there is no ``true'' value of the entropy rate here,
for the purposes of comparison we use the estimate obtained
by the naive CTW estimator (identified as the most effective method
by \cite{gao2008estimating} and \cite{verdu2019empirical})
when all $n=3,919,361$ samples are used, giving
$\bar{H}=0.0241$. 
The resulting estimates for all 
five methods (with the posterior 
mean given for BCT) are summarised in Table~\ref{spike_table}, verifying again that BCT outperforms all the other methods. For the plug-in estimator (and in all subsequent examples), we only show the best block-lengths~$k$.

\vspace*{-0.02 cm}

\begin{table}[!h]
{\small
\begin{tabular}{ccccccc}
\toprule 
 & ``True" & BCT & CTW & PPM &LZ  & $k=5$   \\
\midrule
$\widehat H $ & 0.0241 & \bf{0.0234} & 0.0249 & 0.0360 & 0.0559 & 0.0204  \\
\bottomrule
\end{tabular}
}
\caption{Entropy rate estimates for the neural spike train.}
\label{spike_table}
\vspace*{-0.1 cm}
\end{table}

\noindent
{\bf Financial data. } Here, we consider $n=2000$ observations from the 
financial dataset~F.2 of \cite{our}. This consists of tick-by-tick price 
changes of the Facebook stock price, quantised to three values:
$x_i =0$ if the price goes down, $x_i=1$ if it stays the same,
and $x_i=2$ if it goes up.
The BCT entropy-rate posterior is shown in Figure~\ref{fb_h_hist}:  
It has mean $\mu = 0.921$, and standard deviation $\sigma = 0.028$. 

Once again, as the ``true'' value of the entropy rate we take
the estimate produced by the naive CTW estimator on a 
longer sequence with $n=10^4$ observations,
giving $\bar{H}=0.916$. The results of all five estimators
are summarised in Table~\ref{table_fb},
where for the BCT estimator we once again give the posterior mean, which is again found to outperform the alternatives.

\begin{figure}[!ht]
\begin{subfigure}{0.49 \linewidth}
 \includegraphics[width= 1 \linewidth]{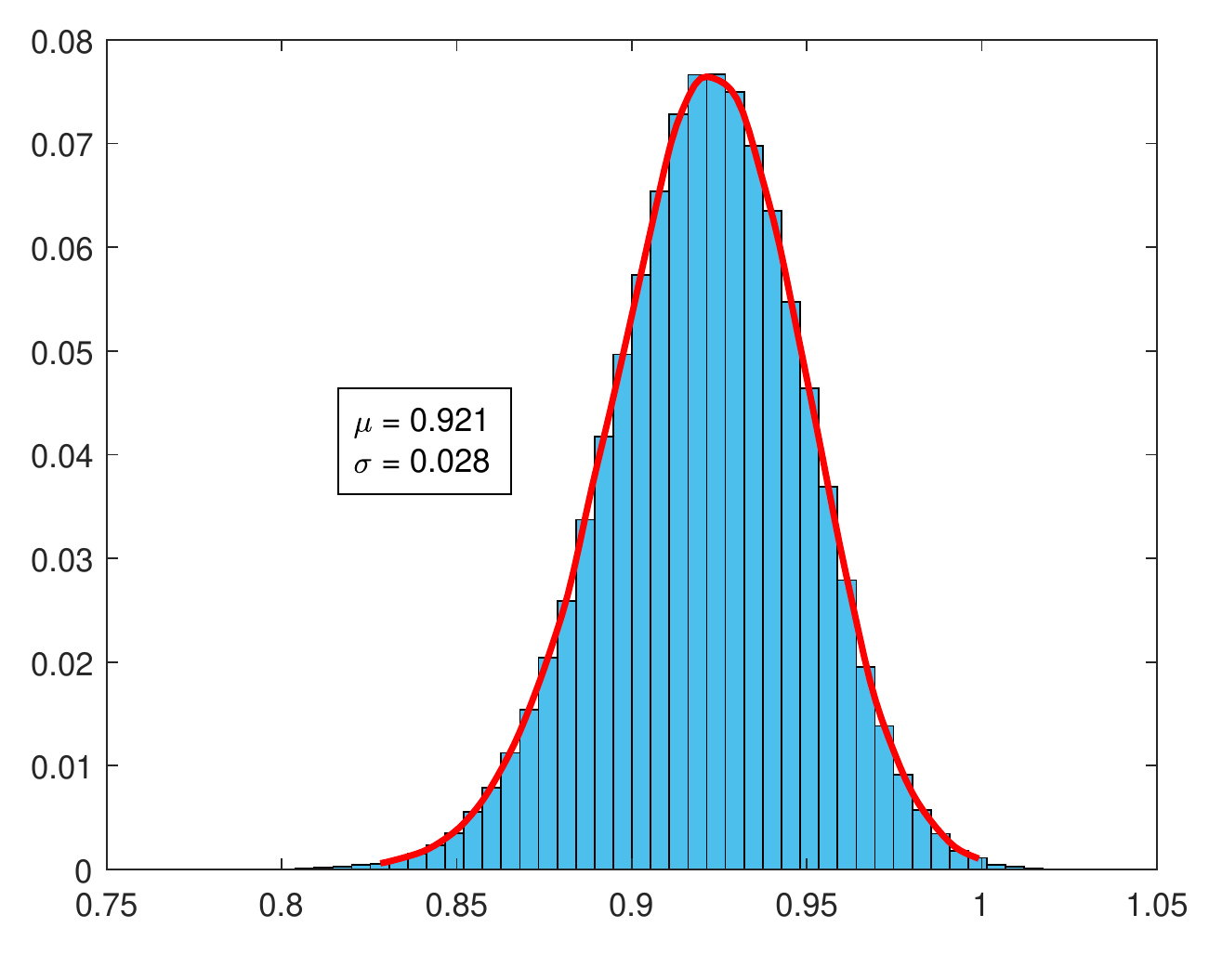}
\vspace*{-0.75 cm}
\caption{financial dataset}
\label{fb_h_hist}
\end{subfigure}
\begin{subfigure}{0.49 \linewidth}
 \includegraphics[width= 1 \linewidth]{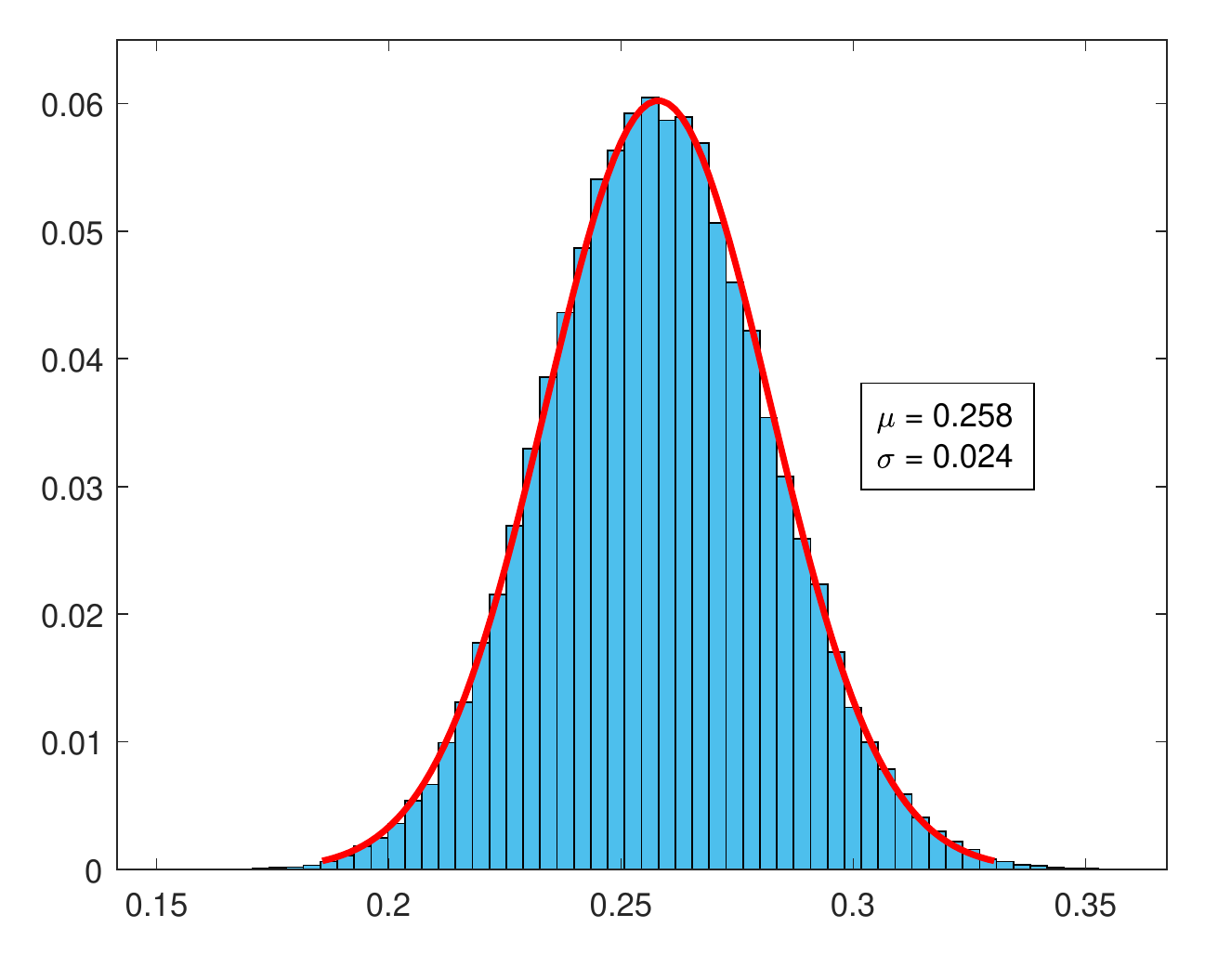}
\vspace*{-0.75 cm}
\caption{pewee birdsong}
\label{pewee_h_hist}
\end{subfigure}
\vspace*{-0.1 cm}
\caption{Histograms of the posterior 
distribution $\pi(\bar{H}|x)$ of the entropy rate,
constructed from $N=10^5$ i.i.d.\ samples.}
\vspace*{-0.0 cm}
\end{figure}

\begin{table}[!ht]
\begin{tabular}{ccccccccc}
\midrule 
 & ``True" & BCT & CTW & PPM &LZ & $k=5$ & $k=6$    \\
\midrule
$\widehat H $ & 0.916 & \bf{0.921} & 0.939 & 1.049 & 0.846 & 0.930 & 0.907   \\
\midrule
\end{tabular}
\vspace*{-0.15 cm}
\caption{Entropy rate estimates for the financial data set.}
\label{table_fb}
\vspace*{-0.5 cm}
\end{table}

\noindent
{\bf Pewee birdsong. } The last data set
examined is a time series $x$ describing
the twilight song of the wood pewee 
bird \cite{craig1943song,sarkar2016bayesian}.
It consists of $n=1327$ observations from an 
alphabet of size $m=3$. 
The BCT posterior is shown in Figure \ref{pewee_h_hist}: 
It is approximately Gaussian with mean $\mu = 0.258$ and standard 
deviation $\sigma = 0.024$. The fact that the standard deviation
is small suggests ``confidence'' in the 
resulting estimates, which is important because here
(as in most real applications) 
there is no knowledge of a ``true'' underlying value. 
Table~\ref{pewee_table} shows all the resulting estimates;
the posterior mean is shown for the BCT estimator.

\begin{table}[!ht]
\begin{tabular}{ccccccccc}
\toprule 
& BCT & CTW & PPM &LZ & $k=5$ & $k=10$ & $k=15$  \\
\midrule
$\widehat H $ & 0.258 &  0.278 & 0.318 & 0.275 & 0.467 & 0.336 & 0.272 \\
\bottomrule
\end{tabular}
\vspace*{-0.05 cm}
\caption{Entropy rate estimates for the pewee song data.}
\label{pewee_table}
\vspace*{-0.35 cm}
\end{table}


\section{Concluding Remarks}

The main conclusion from the results on the six data
sets examined in the experimental section is that the BCT entropy estimator
gives the most accurate and reliable results among the
five estimators considered. In addition to the fact that
the BCT point estimates typically outperform those produced 
by other methods, the BCT estimator is accompanied by the
entire posterior distribution $\pi(\bar{H}|x)$ of the
entropy rate, induced by the observations~$x$. 
As usual, this
distribution can be used to quantify the uncertainty 
in estimating $\bar{H}$, and it contains significantly 
more information than simple point estimates and their
associated confidence intervals.

In closing, we note a few possible directions for future research 
extending this work. First, the proposed entropy estimator could 
be employed for the estimation of mutual information and directed 
information rates between processes, which come with important 
applications, as for example in testing for temporal 
causality~\cite{kontoyiannis2016estimating}. Also, as the 
BCT framework has recently been extended to the setting of 
real-valued time series~\cite{bct_ar}, it would be interesting to 
extend our methods for entropy estimation in the real-valued case.

\newpage



\end{document}